\documentclass[english]{IEEEtran}

\usepackage[T1]{fontenc}
\usepackage[latin9]{inputenc}
\usepackage{verbatim}
\usepackage{float}
\usepackage{amsthm}
\usepackage{amsmath}
\usepackage{amssymb}
\usepackage{graphicx}
\usepackage{color}
\usepackage{url}

\makeatletter

\floatstyle{ruled}
\newfloat{algorithm}{tbp}{loa}
\providecommand{\algorithmname}{Algorithm}
\floatname{algorithm}{\protect\algorithmname}

\numberwithin{equation}{section}
\numberwithin{figure}{section}
\theoremstyle{plain}
\newtheorem{thm}{\protect\theoremname}[section]
\theoremstyle{definition}
\newtheorem{defn}[thm]{\protect\definitionname}
\theoremstyle{remark}

\theoremstyle{plain}
\newtheorem{lem}[thm]{\protect\lemmaname}

\newtheorem*{lem*}{Lemma}

\theoremstyle{remark}
\newtheorem{rem}[thm]{\protect\remarkname}
\theoremstyle{plain}

\theoremstyle{plain}

\@ifundefined{showcaptionsetup}{}{%
 \PassOptionsToPackage{caption=false}{subfig}}
\usepackage{subfig}
\makeatother

\usepackage{babel}
\providecommand{\claimname}{Claim}
\providecommand{\definitionname}{Definition}
\providecommand{\lemmaname}{Lemma}
\providecommand{\remarkname}{Remark}
\providecommand{\theoremname}{Theorem}
\providecommand{\corollaryname}{Corollary}
\providecommand{\propositionname}{Proposition}

\begin{document}


\title{Robust Recovery of Positive Stream of Pulses }


\author{Tamir Bendory \thanks{T. Bendory is with  The Program in Applied and Computational Mathematics, Princeton University, Princeton, NJ, USA   (e-mail:
tamir.bendory@princeton.edu)}}
\maketitle



\begin{abstract}
The problem of estimating the delays and amplitudes of a positive stream of pulses appears in many applications, such as  single-molecule microscopy. 
This paper suggests estimating the delays and amplitudes using a convex program, which is robust in the presence of noise (or model mismatch). Particularly, the recovery error  is proportional to the
noise level. We further show that the error grows exponentially with the density of the delays and also depends on the localization properties
of the pulse.
\end{abstract}



\section{Introduction}

Signals comprised of stream of pulses play a key role in many engineering applications, such as ultrasound imaging and radar
(see, e.g. \cite{Bar-IlanSub-Nyquis,tur2011innovation,wagner2012compressed,vetterli2002sampling,deslauriers2012spherical}).
In some applications, the signal under examination is known to be real and non-negative. 
For instance, in single-molecule microscopy we measure the convolution
of positive point sources with the microscope's point spread function
\cite{klar2000fluorescence,betzig2006imaging,bronstein2009transient}. 
Another example arises from the problem of estimating the
orientations of the white matter fibers in the brain using
diffusion weighted magnetic resonance imaging (MRI). In this application, the measured data is modeled as the convolution of a sparse positive signal on the sphere, which represents the unknown orientations, with a known point spread function that acts as a low-pass filter \cite{tournier2004direct,bendory2015recovery}.

This paper focuses its attention on the model of  \emph{positive stream of pulses}. 
In this model, the measurements are comprised of a sum of unknown shifts of a kernel $\mathbf{g}$ with positive coefficients, i.e.
\begin{equation} \label{eq:1}
\mathbf{y}[k]=\sum_{m}c_{m}\mathbf{g}\left[k-k_{m}\right]+\tilde{\mathbf{n}}[k],\quad k\in\mathbb{Z},\quad c_{m}>0,
\end{equation}
where $\tilde{\mathbf{n}}$ is a bounded error term (noise, model mismatch) obeying $\left\Vert {\tilde{\mathbf{n}}}\right\Vert _{1}:=\sum_{k\in\mathbb{Z}}\left|\tilde{\mathbf{n}}[k]\right|\leq\delta$. We do not assume any prior knowledge on the noise statistics. 
The pulse $\mathbf{g}$ is assumed to be a sampled version of a scaled continuous kernel, namely, $\mathbf{g}[k]:=\mathbf{g}\left(\frac{k}{\sigma N}\right)$, where
$\mathbf{g}(t)$ is the continuous kernel, $\sigma>0$ is a scaling parameter and $1/N$
is the sampling interval. For instance, if $\mathbf{g}$ is Gaussian kernel, then $\sigma$ denotes its standard deviation. The delayed versions of the kernel, $\mathbf{g}\left[k-k_{m}\right]$,
are often referred to as \emph{atoms}. We aim to estimate the set of delays
$\left\{ k_{m}\right\} \subset\mathbb{Z}$ and the positive amplitudes $\left\{ c_{m}>0\right\} $
from the measured data $\mathbf{y}[k]$.

The sought parameters of the stream of pulses model can be defined by a signal of the form  
\begin{equation}
\mathbf{x}[k]:=\sum_{m}c_{m}\boldsymbol{\delta}\left[k-k_{m}\right],\quad c_{m}>0,\label{eq:x}
\end{equation}
where $\boldsymbol{\delta}[k]$ is the one-dimensional Kronecker Delta function 
\begin{equation*}
\boldsymbol{\delta}[k]:=
\begin{cases}
1, & k=0,\\
0 & k\neq 0.
\end{cases}
\end{equation*}
In this manner, the problem can be thought of as a sparse deconvolution problem, namely,
\begin{equation}
\mathbf{y}[k]=\left(\mathbf{g}\ast\left(\mathbf{x}+\mathbf{n}\right)\right)[k],\label{eq:model1}
\end{equation}
where $'\ast'$ denotes a discrete convolution and $ \mathbf{n}[k] $
is the error term.

The one-dimensional model can be extended to higher-dimensions. In this paper we also analyze in detail the model of two-dimensional positive stream of pulses given by
\begin{eqnarray}
\mathbf{y}[\mathbf{k}] & = & \left(\mathbf{g_{2}}\ast\left(\mathbf{x_{2}}+\mathbf{n}\right)\right)\left[\mathbf{k}\right],\label{eq:model2}\\
 & = & \sum_{m}c_{m}\mathbf{g_{2}}\left[\mathbf{k}-\mathbf{k}_{m}\right]+\tilde{\mathbf{n}}[\mathbf{k}],\nonumber 
\end{eqnarray}
where $\mathbf{k}:=\left[k_{1},k_{2}\right]\in\mathbb{Z}^{2}$ and   $\mathbf{g_2}$ is a two-dimensional pulse. As in the one-dimensional case, the pulse is defined as a sampled version of a two-dimensional kernel $\mathbf{g_{2}}(t_1,t_2)$ by $\mathbf{g_{2}}\left[\mathbf{k}\right]=\mathbf{g_2}\left(\frac{k_{1}}{\sigma_{1}N_{1}},\frac{k_{2}}{\sigma_{2}N_{2}}\right)$. The signal
\begin{equation}
\mathbf{x}_{2}[\mathbf{k}]:=\sum_{m}c_{m}\boldsymbol{\delta}\left[\mathbf{k}-\mathbf{k}_{m}\right],\quad c_{m}>0,\label{eq:x2}
\end{equation}
defines the underlying parameters to be estimated, where here $\boldsymbol{\delta}$ denotes the two-dimensional Kronecker Delta function.
For the sake of simplicity, we assume throughout the paper that $\sigma_{1}N_{1}=\sigma_{2}N_{2}:=\sigma N$. 

Many algorithms have been suggested to recover $\mathbf{x}$ from the stream of pulses $\mathbf{y}$. A naive approach would be to estimate $\mathbf{x}$ via least-squares estimation. However, even if the convolution as in (\ref{eq:model1}) is invertible, the  condition number of its associated convolution matrix tends to be extremely high. Therefore, the recovery process is not robust (see for instance section 4.3 in \cite{beck2014introduction}). Suprisngly, the least-squares fails even in a noise-free environment due to amplification of numerical errors. We refer the readers to Figure 1 in \cite{bendorySOP} for a demonstration of this phenomenon. 

A different line of algorithms includes the well-known Prony method, MUSIC, matrix pencil
and ESPRIT, see for  instance \cite{stoica2005spectral,56027,roy1989esprit,schmidt1986multiple,peter2011nonlinear,filbir2012problem,potts2010parameter,potts2013parameter}. 
These algorithms concentrate on estimating the set of delays. Once the set of  delays is known, the coefficients can be easily estimated by least-squares. 
These methods rely on the observation that in  Fourier domain  the stream of pulses model (\ref{eq:1})  reduces to a weighted sum of  complex exponentials, under the assumption that the Fourier transform of $\mathbf{g}$ is non-vanishing. Recent papers analyzed the performance and stability of these algorithms  \cite{liao2014music, liao2014music2, moitra2014threshold,fannjiang2016compressive}. However, as the Fourier transform of the pulse $\mathbf{g}$ tends to be localized and in general contains small values, the stability results do not hold directly  for the stream of pulses model. Furthermore, these methods do not exploit the positivity of the coefficients (if it exists), which is the focus of this work.

In recent years, many convex optimization techniques have been  suggested and analyzed thoroughly for the task super-resolution. Super-resolution is the problem of resolving signals from their noisy low-resolution measurements, see for instance   \cite{candes2013towards,candes2013super,azais2014spike,bhaskar2011atomic,tang2015near,tang2012compressive}.
The main pillar of these works is the duality between robust super-resolution
and the existence of an interpolating polynomial in the measurement
space, called \emph{dual certificate}. Similar techniques have been applied to super-resolve signals on
the sphere \cite{bendory2013exact,bendorySHalgorithm,bendory2015recovery} (see also \cite{filbir2016exact}) and to the recovery of non-uniform splines from their projection onto the space
of low-degree algebraic polynomials \cite{bendory2013Legendre,de2014non}.

The problem of recovering a general signal $\mathbf{x}$ (not necessarily non-negative)
robustly from stream of pulses was considered in \cite{bendorySOP}.
It was shown that the duality between robust recovery and the existence
of an interpolating function holds in this case as well. Particularly,
it turns out that robust recovery is possible if there exists a function, comprised of shifts 
of the kernel $\mathbf{g}$ and its derivatives, that satisfies several interpolation
requirements (see Lemma \ref{lemma:q}). In this case, the solution
of a standard convex program achieves recovery error (in
$\ell_{1}$ norm) of $C^{*}(\mathbf{g})\delta$, for some constant $C^{*}(\mathbf{g})$ that 
depends only on the convolution kernel $\mathbf{g}$. In \cite{SOP_US} it was proven that
the support of the recovered signal is clustered around the support
of the target signal $\mathbf{x}$. The behavior of the solution for large $N$ is analyzed in detail in  \cite{duval2013exact,duval2015sparse}.

The main insight of \cite{bendorySOP} is that the existence of such interpolating function relies on two interrelated pillars.  First, the support of the signal, defined as $\mbox{supp}(\mathbf{x}):=\left\{ t_{m}\right\} =\left\{ k/N\thinspace:\thinspace \mathbf{x}[k]\text{\ensuremath{\neq}}0\right\} $,
should satisfy a \emph{separation condition} of the form
\begin{equation}
\left|t_{i}-t_{j}\right|\geq\nu\sigma,\quad\forall t_{i},t_{j}\in\mbox{supp}(\mathbf{x}),\thinspace i\neq j,\label{eq:1dseparation}
\end{equation}
for some kernel-dependent constant $\nu>0$ which does not depend on
$N$ or $\sigma$. In the two-dimensional case, the separation condition
gets the form %
\footnote{Recall that we assume for simplicity that $\sigma_{1}N_{1}=\sigma_{2}N_{2}:=\sigma N$. %
} 
\begin{equation}
\left\Vert \mathbf{t}_{i}-\mathbf{t}_{j}\right\Vert _{\infty}\geq\nu\sigma,\quad \mathbf{t}_{i},\mathbf{t}_{j}\in\mbox{supp}(\mathbf{x_{2}}),\thinspace i\neq j,\label{eq:2dseparation}
\end{equation}
where $\mathbf{t}_{i}=\left[t_{i,1},t_{i,2}\right]\in\mathbb{R}^2$ and $\left\Vert \mathbf{t}_{i}-\mathbf{t}_{j}\right\Vert _{\infty}:=\max\left\{ \left|t_{i,1}-t_{j,1}\right|,\left|t_{i,2}-t_{j,2}\right|\right\} $.
The second pillar is that the kernel $\mathbf{g}$ would be an \emph{admissible kernel}. An admissible kernel is a function that satisfies some mild localization
properties. These properties  are discussed in the next section (see Definition
\ref{def:kernel}).
Two prime examples for admissible kernels 
are the Gaussian kernel $\mathbf{g}(t)=e^{-\frac{t^{2}}{2}}$ and the Cauchy
kernel $\mathbf{g}(t)=\frac{1}{1+t^{2}}$. In \cite{bendorySOP}, the minimal
separation constant $\nu$ which is required for the existence of
the interpolating function (and hence, for robust recovery) was evaluated
numerically to be 1.1 and 0.5 for the Gaussian and Cauchy kernels, respectively. 

Inspired by the recent work on super-resolution of positive point sources \cite{morgenshtern2014stable}, this work focuses on the model of positive stream of pulses. In contrast to \cite{bendorySOP}, we prove that in this case the separation condition is no longer necessary to achieve stable recovery.  We generalize and improve the results of \cite{morgenshtern2014stable} as discussed in detail in Section \ref{sec:Main-Result}.
Particularly, we show that positive signals of
the form (\ref{eq:x}) can be recovered robustly from the measurements
$\mathbf{y}$ (\ref{eq:model1}) and the recovery error is proportional to the
noise level $\delta$. Furthermore, the recovery error grows exponentially with the density of the signal's support. We characterize the density of the support using the notion of Rayleigh-regularity, which is defined precisely in Section \ref{sec:Main-Result}. The recovery error also depends on the localization properties of the kernel $\mathbf{g}$. A similar result holds for the two-dimensional case.

We use the following notation throughout the paper. We denote an index $k\in\mathbb{Z}$ by brackets $\left[k\right]$
and a continuous variables $t\in\mathbb{R}$ by parenthesis $\left(t\right)$. We use boldface small and capital letters for vectors and matrices, respectively. Calligraphic letters, e.g. $\mathcal{A}$, are used for sets and $\mathcal{\left|A\right|}$ for the cardinality of the set.
The $\ell^{th}$ derivative of $\mathbf{g}(t)$ is denoted as $\mathbf{g}^{(\ell)}(t)$.
For vectors, we use the standard definition of $\ell_{p}$ norm  as ${\left\Vert \mathbf{a}\right\Vert }_{p}:=\left(\sum_{k\in\mathbb{Z}}\left|\mathbf{a}\left[k\right]\right|^{p}\right)^{1/p}$
for $p\geq1$. We reserve $1/N$ to denote the sampling interval of (\ref{eq:1}) and define  the support of the signal $\mathbf{x}$  as $\mbox{supp}(\mathbf{x}):=\left\{ k/N\thinspace:\thinspace \mathbf{x}[k]\text{\ensuremath{\neq}}0\right\} $.
We write $k\in\mbox{supp}(\mathbf{x})$ to denote some $k\in\mathbb{Z}$
satisfying $k/N\in\mbox{supp}(\mathbf{x})$.

The rest of the paper is organized as follows. Section \ref{sec:Main-Result}
presents some basic definitions and states our  main theoretical results. Additionally, we give a detailed comparison with the literature. 
The results are proved in Sections \ref{sec:Proof_main_th} and \ref{sec:Proof_th2}. Section \ref{sec:Numerical-Experiments} shows numerical
experiments, validating the theoretical results. Section
\ref{sec:Conclusions} concludes  the work and aims
to suggest potential extensions.

\section{Main Results \label{sec:Main-Result}}

In \cite{bendorySOP}, it was shown that the underlying one-dimensional signal $\mathbf{x}$ can be recovered  robustly from a stream of pulses $\mathbf{y}$  if its support satisfies a separation condition of the form (\ref{eq:1dseparation}).
Following \cite{morgenshtern2014stable}, this work deals with non-negative signals and shows that in this case 
the separation condition is not necessary. Specifically, we prove that the recovery error depends on the density of the signal's support. This density is defined and quantified by the notion of Rayleigh-regularity. More precisely, a one-dimensional signal with
Rayleigh regularity $r$ has at most $r$ spikes within a resolution
cell:
\begin{defn}
\label{def:rayleigh}We say that the set $\mathcal{P}\subset\left\{ k/N\right\} _{k\in\mathbb{Z}}\subset\mathbb{R}$
is Rayleigh-regular with parameters $(d,r)$ and write $\mathcal{P}\in\mathcal{R}^{idx}(d,r)$
if every interval $(a,b)\subset\mathbb{R}$ of length $\vert a-b\vert=d$
contains no more that $r$ elements of $\mathcal{P}$:
\begin{equation}
\left|\mathcal{P}\cap(a,b)\right|\leq r.
\end{equation}
\end{defn}
Equipped with Definition \ref{def:rayleigh}, we define the sets of
signals 
\begin{eqnarray*}
\mathcal{R}\left(d,r\right) & := & \left\{ \mathbf{x}\thinspace:\thinspace\mbox{supp}(\mathbf{x})\in\mathcal{R}^{idx}(d,r)\right\}. 
\end{eqnarray*}
We further let $\mathcal{R}_+\left(d,r\right) $ be the set of signals in $\mathcal{R}\left(d,r\right) $ with non-negative values.

\begin{rem}
If $r_{1}\leq r_{2}$, then $\mathcal{R}\left(d,r_{1}\right)\subseteq\mathcal{R}\left(d,r_{2}\right)$.
If $d_{1}\leq d_{2}$, then $\mathcal{R}\left(d_{2},r\right)\subseteq\mathcal{R}\left(d_{1},r\right)$.
\end{rem}

Besides the density of the signal's support,  robust estimation of the delays and amplitudes also depends on the convolution kernel $\mathbf{g}$. Particularly,  the kernel should satisfy some mild localization properties. In short, the kernel and its first derivatives should decay sufficiently fast. We say that a kernel $\mathbf{g}$ is \emph{non-negative admissible}
if it meets the following definition:
\begin{defn}
\label{def:kernel}We say that $\mathbf{g}$ is a { non-negative admissible}
kernel if $\mathbf{g}(t)\geq0$ for all $t\in\mathbb{R}$, and:
\begin{enumerate}
\item $\mathbf{g}\in\mathcal{C}^{3}(\mathbb{R})$ and is even.
\item \underline{Global property:} There exist constants $C_{\ell}>0,\ell=0,1,2,3$
such that $\left|\mathbf{g}^{(\ell)}(t)\right|\leq C_{\ell}/\left(1+t^{2}\right)$.
\item \underline{Local property:} There exist constants $\varepsilon,\beta>0$
such that

\begin{enumerate}
\item $\mathbf{g}(t)<\mathbf{g}(\varepsilon)$ for all $\left|t\right|>\varepsilon$. 
\item $\mathbf{g}^{(2)}(t)\leq-\beta$ for all $\left|t\right|\leq\varepsilon$.
\end{enumerate}
\end{enumerate}
\end{defn}

Now, we are ready to state our one-dimensional theorem, which is proved in Section \ref{sec:Proof_main_th}. The theorem states that in the noise free-case, $\delta=0$, a convex program  recovers the delays and amplitudes exactly, for any Rayleigh regularity parameter $r$. Namely, the convolution system is invertible even without any sparsity prior. Additionally, in the presence of noise or model mismatch, the recovery error grows exponentially with $r$ and is proportional to the noise level.

\begin{thm}
\label{th:main}Consider the model (\ref{eq:model1}) for a non-negative
admissible kernel $\mathbf{g}$ as defined in Definition \ref{def:kernel}. Then, there exists $\nu>0$ such that if $\mathcal{\mbox{supp}(\mathbf{x})}\in\mathcal{R}^{idx}\left(\nu\sigma,r\right)$
and $N\sigma>\left(\frac{1}{2}\right)^{\frac{1}{2r}+1}\sqrt{\frac{\beta}{\mathbf{g}(0)}},$
the solution $\hat{\mathbf{x}}$ of the convex problem 
\begin{equation}
\min_{\tilde{\mathbf{x}}}\quad\left\Vert \tilde{\mathbf{x}}\right\Vert _{1}\quad\mbox{subject to}\quad\left\Vert \mathbf{y}-\mathbf{g}\ast\tilde{\mathbf{x}}\right\Vert _{1}\leq\delta,\thinspace\tilde{\mathbf{x}}\geq0,\label{eq:cvx}
\end{equation}
 satisfies 
\begin{equation}
\left\Vert \hat{\mathbf{x}}-\mathbf{x}\right\Vert _{1}\leq C(\mathbf{g},r,\nu)\gamma^{2r}\delta,\label{eq:h_tight}
\end{equation}
where $\gamma:=\max\left\{ N\sigma,\varepsilon^{-1}\right\} $ and
\begin{equation}
\begin{split}
C(\mathbf{g},r,\nu):=&4^{r+1}\left(2^{r}-1\right)\left(\frac{\mathbf{g}\left(0\right)}{\beta}\right)^{r} \left(C_{0}\left(1+\frac{\pi^{2}}{6\nu^{2}}\right)\right)^{r-1}\\ &\cdot \left(\frac{6\nu^{2}}{3\mathbf{g}\left(0\right)\nu^{2}-2\pi^{2}C_{0}}\right)^{r}.\label{eq:C_rg}
\end{split}
\end{equation}
\end{thm}

\begin{rem}
For sufficiently large $\nu$ and $N$, the recovery error can be written as 
\begin{equation*}
\left\Vert \hat{\mathbf{x}}-\mathbf{x}\right\Vert _{1}\leq \tilde{C}(\mathbf{g},r)(N\sigma)^{2r}\delta, 
\end{equation*}
where $\tilde{C}(\mathbf{g},r)$ is a constant that depends only on the kernel $\mathbf{g}$ and $r$.
\end{rem}

In order to extend Theorem \ref{th:main} to the two-dimensional case, we
present the equivalent of Definitions \ref{def:rayleigh} and \ref{def:kernel}
to two-dimensional signals. Notice that the two-dimensional definition of Rayleigh regularity is not a direct extension
of Definition \ref{def:rayleigh} and is quite less intuitive. In order to prove Theorems \ref{th:main}
and \ref{thm:main2}, we assume that the support of the signal  could be
presented as a union of $r$ non-intersecting subsets, which satisfy
the separation conditions of  (\ref{eq:1dseparation}) and (\ref{eq:2dseparation}), respectively. In the one-dimensional case, this property is
implied directly from Definition \ref{def:rayleigh}. However, this
property is not guaranteed by the two-dimensional extension of Definition
\ref{def:rayleigh}. See Figure \ref{fig:counter_exp} for a simple
counter-example. Therefore, in the two-dimensional case the Rayleigh-regularity of a signal is defined as follows:

\begin{defn}
\label{def:2Dray}\cite{morgenshtern2014stable} We say that the set
$\mathcal{P}\subset\left\{ k_{1}/N,k_{2}/N\right\} _{k_{1},k_{2}\in\mathbb{Z}}\subset\mathbb{R}^{2}$
is Rayleigh-regular with parameters $(d,r)$ and write $\mathcal{P}\in\mathcal{R}_{2}^{idx}(d,r)$
if it can be presented as a union of $r$ subsets $\mathcal{P}\mathcal{=}\mathcal{P}_{1}\cup\dots\cup\mathcal{P}_{r}$
that are not intersecting and satisfy the minimum separation constraint
(\ref{eq:2dseparation}). Namely, 
\begin{itemize}
\item for all $1\leq i<j\leq r$, $\mathcal{P}_{i}\cap\mathcal{P}_{j}=\emptyset$,
\item for all $1\leq i\leq r$, $\mathcal{P}{}_{i}$ satisfies: for all
square subsets $\mathcal{D}\in\mathbb{R}^{2}$ of side length $d \times d$,
$\left|\mathcal{P}{}_{i}\cap\mathcal{D}\right|\leq1$. \end{itemize}
\end{defn}

\begin{figure}
\begin{center}
\includegraphics[scale=0.4]{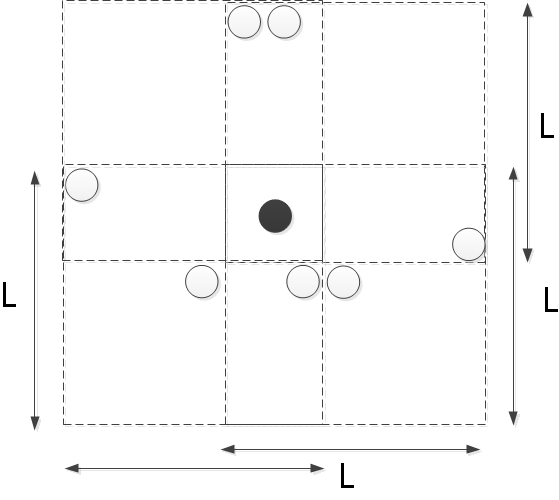}\end{center}\caption{\label{fig:counter_exp}
This figure presents an example for a set of 8 points that cannot be decomposed into 4 non-intersecting
subsets that meet the separation condition as in (\ref{eq:2dseparation}), although each resolution cell contains at most 4 points.  Consider a square resolution cell (in $\ell_\infty$ sense) of size $L\times L$. Note that indeed each resolution cell contains at most 4 points. Nonetheless,  this set of points cannot be described as 4 non-intersecting sets that satisfy the separation condition. Specifically, the distance (in $\ell_{\infty}$ norm) of the black point is smaller
than $L$ from the other 7 points, meaning it has to be in a separate subset from the others.
On the other hand, there is no
triplet of points that comprises a legal subset. Therefore, the property of
Definition \ref{def:2Dray} is not a consequence of the two-dimensional
version of Definition \ref{def:rayleigh}. }
\end{figure}

A two-dimensional non-negative admissible kernel is defined as follows:
\begin{defn}
\label{def:kernel2d}We say that $\mathbf{g_{2}}$ is a two-dimensional {non-negative
admissible }kernel if $\mathbf{g_{2}}\left(t_{1},t_{2}\right)\geq0$ for all
$\left(t_{1},t_{2}\right)\in\mathbb{R}^{2}$ and it has the following
properties:
\begin{enumerate}
\item $\mathbf{g_{2}}\in\mathcal{C}^{3}\left(\mathbb{R}^{2}\right)$ and 
\[
\mathbf{g_{2}}\left(t_{1},t_{2}\right)=\mathbf{g_{2}}\left(-t_{1},t_{2}\right)=\mathbf{g_{2}}\left(t_{1},-t_{2}\right)=\mathbf{g_{2}}\left(-t_{1},-t_{2}\right).
\]

\item \underline{Global property}: There exist constants $C_{\ell_{1},\ell_{2}}>0$
such that $\left|\mathbf{g_{2}}^{(\ell_{1},\ell_{2})}\left(t_{1},t_{2}\right)\right|\le\frac{C_{\ell_{1},\ell_{2}}}{\left(1+t_{1}^{2}+t_{2}^{2}\right)^{3/2}}$,
for $\ell_{1}+\ell_{2}\leq3$, where $\mathbf{g_{2}}^{(\ell_{1},\ell_{2})}\left(t_{1},t_{2}\right):=\frac{\partial^{\ell_{1}}\partial^{\ell_{2}}}{\partial t_{1}^{\ell_{1}}\partial t_{2}^{\ell_{2}}}\mathbf{g_{2}}\left(t_{1},t_{2}\right)$.
\item \underline{Local property}: There exist constants $\beta,\varepsilon>0$
such that

\begin{enumerate}
\item $\mathbf{g_{2}}\left(t_{1},t_{2}\right)<\mathbf{g_{2}}(\varepsilon,0)$ for all $\left(t_{1},t_{2}\right)$
satisfying $|t_{1}|>\varepsilon$, and $\mathbf{g_{2}}\left(t_{1},t_{2}\right)<\mathbf{g_{2}}(0,\varepsilon)$
for all $\left(t_{1},t_{2}\right)$ satisfying $|t_{2}|>\varepsilon$.
\item $\mathbf{g_{2}}^{(2,0)}\left(t_{1},t_{2}\right),\mathbf{g_{2}}^{(0,2)}\left(t_{1},t_{2}\right)<-\beta$
for all $\left(t_{1},t_{2}\right)$ satisfying $|t_{1}|,|t_{2}|\leq\epsilon$. 
\end{enumerate}
\end{enumerate}
\end{defn}

Equipped with the appropriate definitions of Rayleigh regularity and non-negative admissible kernel, we are ready to state our main theorem for the two-dimensional case. The theorem is  proved in Section \ref{sec:Proof_th2}.
\begin{thm}
\label{thm:main2} Consider the model (\ref{eq:model2}) for a non-negative
two-dimensional admissible kernel $\mathbf{g_{2}}$ as defined in Definition \ref{def:kernel2d}. Then, there exists $\nu>0$
such that if $\mathcal{\mbox{supp\ensuremath{\left(\ensuremath{\mathbf{x_{2}}}\right)}}}\in\mathcal{R}_{2}^{idx}\left(\nu\sigma,r\right)$,
the solution $\hat{\mathbf{x}}_\mathbf{2}$ of the convex problem 
\begin{equation}
\min_{\tilde{\mathbf{x}}}\quad\left\Vert \tilde{\mathbf{x}}\right\Vert _{1}\quad\mbox{subject to}\quad\left\Vert \mathbf{y}-\mathbf{g_{2}}\ast\tilde{\mathbf{x}}\right\Vert_1 \leq\delta,\thinspace\tilde{\mathbf{x}}\geq0,\label{eq:cvx2}
\end{equation}
 satisfies (for sufficiently large $N$ and $\nu$) 
\[
\left\Vert \hat{\mathbf{x}}_{\mathbf{2}}-\mathbf{x_{2}}\right\Vert _{1}\leq C_{2}\left(\mathbf{g_2},r\right)\left(N\sigma\right)^{2r}\delta,
\]
 where $C_{2}\left(\mathbf{g_2},r\right)$ is a constant which depends on the
kernel $\mathbf{g_2}$ and the Rayleigh regularity $r$.
\end{thm}

To conclude this section, we summarize the contribution of this paper and compare it to the relevant previous works. Particularly, we stress the chief differences from \cite{morgenshtern2014stable,bendorySOP} which served as inspiration for this work.

\begin{itemize}
\item This work deviates from \cite{morgenshtern2014stable} in two important aspects. First, our stability results is much stronger than those in \cite{morgenshtern2014stable}. Particularly, our main results hold for signals with $r$ spikes within a resolution cell. In contrast, the main theorems of   \cite{morgenshtern2014stable} require signals with  $r$ spikes within $r$ resolution cells. Second, our formulation is not restricted to kernels with finite bandwidth and, in this manner, can be seen as a generalization of  \cite{morgenshtern2014stable}. This generalization is of particular interest as many kernels in practical applications are not band-limited.

\item In \cite{bendorySOP}, it is proven that robust recovery from general stream of pulses (not necessarily non-negative) is possible if the delays are not clustered. Here, we show that the separation is unnecessary in the positive case and can be replaced by the notion of Raleigh regularity. This notion quantifies the density of the signal's support.  

\item We derive strong stability guarantees compared to  parametric methods, such as Prony method, matrix pencil and MUSIC. Nevertheless, we heavily rely  on the positiveness of signal and the density of the delays,
whereas the parametric methods do not have these restrictions. We
also mention that several previous works suggested noise-free results
for non-negative signals in similar settings, however they do not
derive stability guarantees \cite{bendory2013exact,de2012exact,fuchs2005sparsity,schiebinger2015superresolution}. In \cite{denoyelle2015support} it was proven that the necessary separation between the delays drops to zero for sufficiently low noise level.

\end{itemize}

\section{\label{sec:Proof_main_th}Proof of Theorem \ref{th:main}}

The proof follows the outline of \cite{morgenshtern2014stable} and borrows  constructions from \cite{bendorySOP}.
Let $\hat{\mathbf{x}}$ be s solution of (\ref{eq:cvx}) and set $\mathbf{h}:=\hat{\mathbf{x}}-\mathbf{x}$.
Observe that by (\ref{eq:cvx}) $\left\Vert \mathbf{h}\right\Vert _{1}$ is
finite since $\left\Vert \mathbf{h}\right\Vert _{1}\leq\left\Vert \hat{\mathbf{x}}\right\Vert _{1}+\left\Vert \mathbf{x}\right\Vert _{1}\leq2\left\Vert \mathbf{x}\right\Vert _{1}$.
The proof relies on some fundamental results from \cite{bendorySOP}
(particularly, see Proposition 3.3 and Lemmas 3.4 and 3.5) which are
summarized by the following lemma: 
\begin{lem}
\label{lemma:q} Let $\mathbf{g}$ be a non-negative admissible kernel as defined in Definition \ref{def:kernel} and
suppose that $\mathcal{T}:=\left\{ t_{m}\right\} \in\mathcal{R}^{idx}\left(\nu\sigma,1\right)$.
Then, there exists a kernel-dependent separation constant $\nu>0$ (see (\ref{eq:1dseparation}))
and a set of coefficients $\left\{ a_{m}\right\} $ and $\left\{ b_{m}\right\} $
such there exists an associated function of the form 
\begin{equation}
\tilde{\mathbf{q}}(t)=\sum_{m}a_{m}\mathbf{g}\left(\frac{t-t_{m}}{\sigma}\right)+b_{m}\mathbf{g}^{(1)}\left(\frac{t-t_{m}}{\sigma}\right),\label{eq:q_form}
\end{equation}
which satisfies: 
\begin{eqnarray*}
\tilde{\mathbf{q}}\left(t_{m}\right) & = & 1,\quad t_{m}\in\mathcal{T},\\
\tilde{\mathbf{q}}\left(t\right) & \leq & 1-\frac{\beta\left(t-t_{m}\right)^{2}}{4\mathbf{g}(0)\sigma^{2}},\quad\left|t-t_{m}\right|\leq\varepsilon\sigma,\thinspace t_{m}\in\mathcal{T},\\
\tilde{\mathbf{q}}\left(t\right) & < & 1-\frac{\beta\varepsilon^{2}}{4\mathbf{g}(0)},\quad\left|t-t_{m}\right|>\varepsilon\sigma,\thinspace\forall t_{m}\in\mathcal{T},\\
\tilde{\mathbf{q}}\left(t\right) & \geq & 0  ,\quad t\in\mathbb{R},
\end{eqnarray*}
where  $\varepsilon$ and $\beta$ are the constants associated with 
$\mathbf{g}$. Furthermore,
\begin{eqnarray}
\left\Vert \mathbf{a}\right\Vert _{\infty} & := & \underset{m}{\max}\left\vert a_{m}\right\vert \leq\frac{3\nu^{2}}{3\mathbf{g}\left(0\right)\nu^{2}-2\pi^{2}C_{0}},\label{eq:a}\\ \label{eq:b}
\left\Vert \mathbf{b}\right\Vert _{\infty} & := & \underset{m}{\max}\left\vert b_{m}\right\vert \\ &\leq & \frac{\pi^{2}C_{1}}{\left(3\left\vert \mathbf{g}^{\left(2\right)}\left(0\right)\right\vert \nu^{2}-\pi^{2}C_{2}\right)\left(3\mathbf{g}\left(0\right)\nu^{2}-2\pi^{2}C_{0}\right)}\nonumber.
\end{eqnarray}
\end{lem}
\begin{rem}
The non-negativity property, $\tilde{\mathbf{q}}\left(t\right)\geq0$ for all
$t\in\mathbb{R}$, does not appear in \cite{bendorySOP}, however,
it is a direct corollary of the non-negativity assumption that $\mathbf{g}(t)\geq0$
for all $t\in\mathbb{R}$.
\end{rem}

The interpolating function (\ref{eq:q_form}) also satisfies the following
property which will be needed in the proof:
\begin{lem}
\label{claim:qh} Let $\hat{\mathbf{x}}$ be s solution of (\ref{eq:cvx}) and set $\mathbf{h}:=\hat{\mathbf{x}}-\mathbf{x}$. Let $\left\{ \mathcal{T}_{i}\right\} _{i=1}^{r}$ be a  union of $r$ non-intersecting
sets  obeying $\mathcal{T}_{i}\in\mathcal{R}^{idx}\left(\nu\sigma,1\right)$
for all $i\in\left\{ 1,\dots,r\right\} $. For each set $\mathcal{T}_{i}$,
let $\tilde{\mathbf{q}}_{i}[k]:=\tilde{\mathbf{q}}_{i}\left(k/N\right)$, $k\in\mathbb{Z}$,
be an associated function, where $\tilde{\mathbf{q}_{i}}(t)$ is given in (\ref{eq:q_form}).
Then, for any sequence $\left\{ \alpha_{i}\right\} _{i=1}^{r}\in\left\{ 0,1\right\} $
we have 
\begin{equation} \label{eq:lemma3.4}
\begin{split}
&\sum_{k\in\mathbb{Z}}\prod_{i=1}^{r}\left(\tilde{\mathbf{q}}_{i}\left[k\right]\right)^{\alpha_{i}}\mathbf{h}[k]\\ &\leq\left(C_{0}\left(1+\frac{\pi^{2}}{6\nu^{2}}\right)\right)^{r-1}\left(\frac{6\nu^{2}}{3\mathbf{g}\left(0\right)\nu^{2}-2\pi^{2}C_{0}}\right)^{r}\delta \\ &+c^{*}\nu^{-4}\left\Vert \mathbf{h}\right\Vert _{1},
\end{split}
\end{equation}
for some constant $c^{*}>0$ that depends on the kernel $\mathbf{g}$.
\end{lem}
\begin{proof}
We begin by two preliminary calculations. First, we observe from (\ref{eq:model1})
and (\ref{eq:cvx}) that 
\begin{eqnarray}
\sum_{n\in\mathbb{Z}}\left|\sum_{k\in\mathbb{Z}}\mathbf{g}\left[k-n\right]\mathbf{h}[k]\right| & \leq & \left\Vert \mathbf{y}-\mathbf{g}\ast\hat{\mathbf{x}}\right\Vert _{1} \nonumber \\ &+&\left\Vert \mathbf{g}\ast \mathbf{x}-\mathbf{y}\right\Vert _{1} \nonumber\\ &\leq & 2\delta.\label{eq:delta}
\end{eqnarray}
Additionally, we can estimate for all $k\in\mathbb{Z}$ (see Section
3.4 in \cite{bendorySOP}) 
\[
\sum_{k_{m}\in\mathcal{T}_{i}}\frac{1}{1+\left(\frac{k-k_{m}}{N\sigma}\right)^2}<2\left(1+\frac{\pi^{2}}{6\nu^{2}}\right),
\]
and hence with the properties of admissible kernel as defined in Definition
\ref{def:kernel} we have for $\ell=0,1,$ 
\begin{eqnarray}
\left|\sum_{k_{m}\in\mathcal{T}_{i}}\mathbf{g}^{(\ell)}\left[k-k_{m}\right]\right| & \leq & C_{\ell}\sum_{k_{m}\in\mathcal{T}_{i}}\frac{1}{1+\left(\frac{k-k_{m}}{N\sigma}\right)^2}\nonumber \\ &\leq & 2C_{\ell}\left(1+\frac{\pi^{2}}{6\nu^{2}}\right).\label{eq:bound_g1h}
\end{eqnarray}
According to (\ref{eq:q_form}), the left-hand of (\ref{eq:lemma3.4})
can be explicitly written as: 
\begin{equation}
\begin{split}
&\sum_{k\in\mathbb{Z}}\prod_{i=1}^{r}\left(\tilde{\mathbf{q}}_{i}\left[k\right]\right)^{\alpha_{i}}\mathbf{h}[k] \\ & =\sum_{k\in\mathbb{Z}}\prod_{i=1}^{r}\left(\sum_{k_{m_{i}}\in\mathcal{T}_{i}}a_{m_{i}}\mathbf{g}\left[k-k_{m_{i}}\right]+b_{m_{i}}\mathbf{g}^{(1)}\left[k-k_{m_{i}}\right]\right)^{\alpha_{i}}\mathbf{h}[k].\label{eq:qh_explicit}
\end{split}
\end{equation}
This expression can be decomposed into (at most) $2^{r}$ terms. We
commence by considering the first term of the expression with $\alpha_{1}=\alpha_{2}=1$
and $\alpha_{i}=0$ for $i>2$ (namely, the product of the shifts of $\mathbf{g}$). Using (\ref{eq:delta})
and (\ref{eq:bound_g1h}) we get 
\begin{eqnarray*}
&&\sum_{k\in\mathbb{Z}}\sum_{k_{m_{1}}\in\mathcal{T}_{1}}a_{m_{1}}\mathbf{g}\left[k-k_{m_{1}}\right]\sum_{k_{m_{2}}\in\mathcal{T}_{2}}a_{m_{2}}\mathbf{g}\left[k-k_{m_{2}}\right]\mathbf{h}[k] \\ & \leq & \left\Vert \mathbf{a}\right\Vert _{\infty}^{2}\left|\sum_{_{k_{m_{1}}\in\mathcal{T}_{1}}}\sum_{k\in\mathbb{Z}}\mathbf{g}\left[k-k_{m_{1}}\right]\mathbf{h}[k]\sum_{k_{m_{2}}\in\mathcal{T}_{2}}\mathbf{g}\left[k-k_{m_{2}}\right]\right|\\
 & \leq & 2\left\Vert \mathbf{a}\right\Vert _{\infty}^{2}C_{0}\left(1+\frac{\pi^{2}}{6\nu^{2}}\right)\sum_{n\in\mathbb{Z}}\left|\sum_{k\in\mathbb{Z}}\mathbf{g}\left[k-n\right]\mathbf{h}[k]\right|\\
 & \leq & 4\left\Vert \mathbf{a}\right\Vert _{\infty}^{2}C_{0}\left(1+\frac{\pi^{2}}{6\nu^{2}}\right)\delta.
\end{eqnarray*}
From the same methodology and using (\ref{eq:a}), we conclude that
for any sequence of coefficients $\left\{ \alpha_{i}\right\} _{i=1}^{r}\in\left\{ 0,1\right\} $
we get 
\begin{equation*}
\begin{split}
&\sum_{k\in\mathbb{Z}}\prod_{i=1}^{r}\left(\sum_{k_{m_{i}}\in\mathcal{T}_{i}}a_{m_{i}}\mathbf{g}\left[k-k_{m_{i}}\right]\right)\mathbf{h}[k] \\ &\leq\left(C_{0}\left(1+\frac{\pi^{2}}{6\nu^{2}}\right)\right)^{r-1}\left(\frac{6\nu^{2}}{3\mathbf{g}\left(0\right)\nu^{2}-2\pi^{2}C_{0}}\right)^{r}\delta.
\end{split}
\end{equation*}

Next, using (\ref{eq:bound_g1h}) we observe that all other $2^{r}-1$
terms of (\ref{eq:qh_explicit}) can be bounded by $c_{0}\left\Vert \mathbf{a}\right\Vert _{\infty}^{\beta_{1}}\left\Vert \mathbf{b}\right\Vert _{\infty}^{\beta_{2}}\left\Vert \mathbf{h}\right\Vert _{1}$
for some constant $c_{0}>0$ and $0\leq\beta_{1}\leq r-1$, $1\leq\beta_{2}\leq r$.
Hence, we conclude by (\ref{eq:a}) and (\ref{eq:b}) that the summation
of all these terms is bounded by $c^{*}\nu^{-4}\left\Vert \mathbf{h}\right\Vert _{1}$
for  sufficiently large constants $c^{*}>0$ and  $\nu$. The constant $c^{*}>0$  depends only on the kernel $\mathbf{g}$. This completes
the proof.
\end{proof}
Consider $\mathbf{x}\in\mathcal{R}_{+}\left(\nu\sigma,r\right)$ and let us
define the sets $\mathcal{N}:=\left\{ k/N\thinspace:\thinspace \mathbf{h}[k]<0\right\} $
and respectively $\mathcal{N}^{C}:=\left\{ k/N\thinspace:\thinspace \mathbf{h}[k]\geq0\right\} $.
Throughout the proof, we use the notation $k\in\mathcal{N}$ and
$k\in\mathcal{N}^{C}$ to denote some $k\in\mathbb{Z}$ so that $k/N\in\mathcal{N}$
and $k/N\in\mathcal{N}^{C}$, respectively. Observe that by definition,
$\mathcal{N}\subseteq\mbox{supp}(\mathbf{x})$ and thus $\mathcal{N}\in\mathcal{R}^{idx}\left(\nu\sigma,r\right)$.
The set $\mathcal{N}$ can be presented as the union of $r$ non-intersecting
subsets $\mathcal{N}=\cup_{i=1}^{r}\mathcal{N}_{i}$, where $\mathcal{N}_{i}=\left\{ t_{i},t_{i+r},t_{i+2r},\dots\right\} $
and $\mathcal{N}_{i}\in\mathcal{R}^{idx}\left(\nu\sigma,1\right)$.
Therefore, for each subset $\mathcal{N}_{i}$ there exists an associated
function $\tilde{\mathbf{q}}_{i}[k]=\tilde{\mathbf{q}}_{i}\left(k/N\right)$ as given
in Lemma \ref{lemma:q}. The proof builds upon the following construction 
\begin{equation}
\mathbf{q}[k]:=\prod_{i=1}^{r}\left(1-\tilde{\mathbf{q}}_{i}[k]\right)-\rho,\label{eq:q}
\end{equation}
for some constant $\rho>0$, to be defined later. The function $\mathbf{q}[k]$ satisfies the following properties:
\begin{lem}
\label{lem:q_prod}Let $\mathbf{q}$ be as in (\ref{eq:q}), let $N\sigma>\left(\frac{1}{2}\right)^{\frac{1}{2r}+1}\sqrt{\frac{\beta}{\mathbf{g}(0)}}$ and let 
\begin{equation}
\rho\geq\frac{1}{2}\left(\frac{\beta}{4\mathbf{g}(0)\gamma^{2}}\right)^{r},\label{eq:rho}
\end{equation}
where $\gamma:=\max\left\{ N\sigma,\varepsilon^{-1}\right\} $. Then,
we have 
\begin{eqnarray*}
\mathbf{q}\left[k_{m}\right] & = & -\rho,\quad k_{m}\in\mathcal{N},\\
\mathbf{q}\left[k\right] & \geq & \rho,\quad k\in\mathcal{N}^{C},\\
\mathbf{q}\left[k\right] & \leq & 1,\quad k\in\mathbb{Z}.
\end{eqnarray*}
\end{lem}
\begin{proof}
Since $\mathcal{N}_{i}\in\mathcal{R}^{idx}\left(\nu\sigma,1\right),$
by Lemma \ref{lemma:q} there exists for each subset $\mathcal{N}_{i}$
an associated interpolating function $\tilde{\mathbf{q}_{i}}[k]=\tilde{\mathbf{q}_{i}}(k/N)$.
Consequently, for all $k_{m}\in\mathcal{N}$ we obtain 
\begin{equation*}
\begin{split}
\mathbf{q}\left[k_{m}\right]&=\prod_{i=1}^{r}\left(1-\tilde{\mathbf{q}}_{i}\left[k_{m}\right]\right)-\rho \\
&=-\rho,
\end{split}
\end{equation*}
and for all $k\in\mathcal{N}^{C}$ we have 
\begin{eqnarray*}
\mathbf{q}[k] & = & \prod_{i=1}^{r}\left(1-\tilde{\mathbf{q}}_{i}\left[k\right]\right)-\rho\\ &\geq & \left(\frac{\beta}{4\mathbf{g}(0)\gamma^{2}}\right)^{r}-\rho.
\end{eqnarray*}
By setting 
\[
\rho:=\arg\min_{k\in\mathcal{N}^{C}}\mathbf{q}[k]\geq\frac{1}{2}\left(\frac{\beta}{4\mathbf{g}(0)\gamma^{2}}\right)^{r},
\]
we conclude the proof. Note that in order to guarantees $\rho<1$,
we require $N\sigma>\left(\frac{1}{2}\right)^{\frac{1}{2r}+1}\sqrt{\frac{\beta}{\mathbf{g}(0)}}$
.
\end{proof}
Equipped with Lemma \ref{lem:q_prod}, we conclude that $\mathbf{q}[k]$ and
$\mathbf{h}[k]$ have the same sign for all $k\in\mathbb{Z}$, and thus 
\begin{equation}
\begin{split}
\left\langle \mathbf{q},\mathbf{h}\right\rangle =\sum_{k\in\mathbb{Z}}\mathbf{q}[k]\mathbf{h}[k]&=\sum_{k\in\mathbb{Z}}\vert \mathbf{q}[k]\vert \vert \mathbf{h}[k]\vert \\ &\geq\rho\left\Vert\mathbf{ h}\right\Vert _{1}.\label{eq:low-bound}
\end{split}
\end{equation}
To complete the proof, we need to bound the inner product $\left\langle \mathbf{q},\mathbf{h}\right\rangle $
from above. To this end, observe that 
\begin{equation}
\prod_{i=1}^{r}\left(1-\tilde{\mathbf{q}}_{i}\left[k\right]\right)=1+\boldsymbol{\kappa}_{r}[k],\label{eq:1-q}
\end{equation}
where 
\begin{eqnarray}
\boldsymbol{\kappa}_{r}[k] & := & \sum_{j=1}^{2^{r}-1}\prod_{i=1}^{r}\left(-\tilde{\mathbf{q}}_{i}\left[k\right]\right)^{\alpha_{i}(j)},\label{eq:Q}
\end{eqnarray}
for some coefficients $\left\{ \alpha_{i}(j)\right\} _{i=1}^{r}\in\left\{ 0,1\right\} $.
For instance, $\boldsymbol{\kappa}_{2}[k]=-\tilde{\mathbf{q}}_{1}\left[k\right]-\tilde{\mathbf{q}}_{2}\left[k\right]+\tilde{\mathbf{q}}_{1}\left[k\right]\tilde{\mathbf{q}}_{2}\left[k\right].$
Therefore, by (\ref{eq:q}) and (\ref{eq:1-q}) we get

\begin{eqnarray}
\left\langle \mathbf{q},\mathbf{h}\right\rangle  & = & \left\langle \prod_{i=1}^{r}\left(1-\tilde{\mathbf{q}}_{i}\left[k\right]\right)-\rho,\mathbf{h}\right\rangle \label{eq:qh1}\\
 & = & \left\langle \left(1-\rho\right)+\boldsymbol{\kappa}_{r},\mathbf{h}\right\rangle \nonumber \\
 & = & \left(1-\rho\right)\sum_{k\in\mathbb{Z}}\mathbf{h}[k]+\left\langle \boldsymbol{\kappa}_{r},\mathbf{h}\right\rangle .\nonumber 
\end{eqnarray}

Recall that by (\ref{eq:cvx}) we have $\left\Vert \hat{\mathbf{x}}\right\Vert _{1}\leq\left\Vert \mathbf{x}\right\Vert _{1}$
and therefore 
\begin{equation*}
\begin{split}
\left\Vert \mathbf{x}\right\Vert _{1}\geq\left\Vert \mathbf{x}+\mathbf{h}\right\Vert _{1}&=\sum_{k\in\mbox{supp}(\mathbf{x})}\left|\mathbf{x}[k]+\mathbf{h}[k]\right|\\ & +\sum_{k\in\mathbb{Z}\backslash\mbox{supp}(\mathbf{x})}\left|\mathbf{h}[k]\right|.
\end{split}
\end{equation*}
By definition $\mathbf{h}[k]\geq0$ for all $k\in\mathcal{N}^{C}$ and we use
the triangle inequality to deduce 
\begin{eqnarray*}
\left\Vert \mathbf{x}\right\Vert _{1} & \geq & \sum_{k\in\mathbb{Z}\backslash\mbox{supp}(\mathbf{x})}\mathbf{h}[k]\\ &+&\sum_{k\in\mbox{supp}(\mathbf{x})\backslash\mathcal{N}}\left(\mathbf{x}[k]+\mathbf{h}[k]\right)+\sum_{k\in\mathcal{N}}\left|\mathbf{x}[k]+\mathbf{h}[k]\right|\\
 & \geq & \left\Vert \mathbf{x}\right\Vert _{1}+\sum_{k\in\mathcal{N}^{C}}\mathbf{h}[k]-\sum_{k\in\mathcal{N}}\left|\mathbf{h}[k]\right|,
\end{eqnarray*}
and thus we conclude
\begin{equation}
\sum_{k\in\mathbb{Z}}\mathbf{h}[k]\leq0.\label{eq:nsp}
\end{equation}
So, from (\ref{eq:Q}), (\ref{eq:qh1}), (\ref{eq:nsp}) and Lemma
\ref{claim:qh} we conclude that 
\begin{eqnarray}
\left\langle \mathbf{q},\mathbf{h}\right\rangle  & \leq & \left|\left\langle \boldsymbol{\boldsymbol{\kappa}}_{r},\mathbf{h}\right\rangle \right|\leq\sum_{j=1}^{2^{r}-1}\left|\sum_{k\in\mathbb{Z}}\prod_{i=1}^{r}\left(\tilde{\mathbf{q}}_{i}\left[k\right]\right)^{\alpha_{i}(j)}\mathbf{h}[k]\right|\nonumber \\
 & \leq & \left(2^{r}-1\right)\left(C_{0}\left(1+\frac{\pi^{2}}{6\nu^{2}}\right)\right)^{r-1}\nonumber\\ &\cdot &\left(\frac{6\nu^{2}}{3\mathbf{g}\left(0\right)\nu^{2}-2\pi^{2}C_{0}}\right)^{r}\delta \nonumber \\ &+&c^{*}\left(2^{r}-1\right)\nu^{-4}\left\Vert \mathbf{h}\right\Vert _{1}. \label{eq:upper-bound}
\end{eqnarray}
Combining (\ref{eq:upper-bound}) with (\ref{eq:low-bound}) and (\ref{eq:rho})
yields 
\begin{eqnarray*}
\left\Vert \mathbf{h}\right\Vert _{1} & \leq & \frac{\left(2^{r}-1\right)\left(C_{0}\left(1+\frac{\pi^{2}}{6\nu^{2}}\right)\right)^{r-1}\left(\frac{6\nu^{2}}{3\mathbf{g}\left(0\right)\nu^{2}-2\pi^{2}C_{0}}\right)^{r}}{\frac{1}{2}\left(\frac{\beta}{4\mathbf{g}(0)\gamma^{2}}\right)^{r}-c^{*}\left(2^{r}-1\right)\nu^{-4}}\delta.
\end{eqnarray*}
This completes the proof of Theorem \ref{th:main}.

\section{\label{sec:Proof_th2}Proof of Theorem \ref{thm:main2}}

The proof of Theorem \ref{thm:main2} follows the methodology of the
proof in Section \ref{sec:Proof_main_th}. We commence by stating
the extension of Lemma \ref{lemma:q} to the two-dimensional case, based
on results from \cite{bendorySOP}:
\begin{lem}
\label{lemma:q2} Let $\mathbf{g_{2}}$ be a non-negative two-dimensional admissible
kernel as defined in Definition \ref{def:kernel} and suppose that
$\mathcal{T}_{2}:=\left\{ \mathbf{t}_{m}\right\} \in\mathcal{R}_{2}^{idx}\left(\nu\sigma,1\right)$.
Then, there exists a kernel-dependent separation constant $\nu>0$
and a set of coefficients $\left\{ a_{m}\right\} ,\left\{ b_{m}^{1}\right\} $
and $\left\{ b_{m}^{2}\right\} $ such that there exist an associated function of the form
\begin{equation}
\begin{split}
\mathbf{q_{2}}(\mathbf{t})=\sum_{m}a_{m}\mathbf{g_{2}}\left(\frac{\mathbf{t}-\mathbf{t}_{m}}{\sigma}\right)&+b_{m}^{1}\mathbf{g_{2}}^{(1,0)}\left(\frac{\mathbf{t}-\mathbf{t}_{m}}{\sigma}\right)\\&+b_{m}^{2}\mathbf{g_{2}}^{(0,1)}\left(\frac{\mathbf{t}-\mathbf{t}_{m}}{\sigma}\right),\label{eq:q2_form}
\end{split}
\end{equation}
 which satisfies:
\begin{eqnarray*}
\tilde{\mathbf{q}}_\mathbf{2}\left(\mathbf{t}\right) & = & 1,\quad{t}_{m}\in\mathcal{T}_{2},\\
\tilde{\mathbf{q}}\left(\mathbf{t}\right) & \leq & 1-c_{1}\frac{\left\Vert \mathbf{t}-\mathbf{t}_{m}\right\Vert _{2}^{2}}{\sigma^{2}},\left\Vert \mathbf{t}-\mathbf{t}_{m}\right\Vert _{\infty}\leq\sigma\varepsilon_{1},\thinspace\mathbf{t}_{m}\in\mathcal{T}_{2},\\
\tilde{\mathbf{q}}\left(\mathbf{t}\right) & \leq & 1-c_{2},\quad\left\Vert \mathbf{t}-\mathbf{t}_{m}\right\Vert _{\infty}>\varepsilon_{1}\sigma,\thinspace\forall \mathbf{t}_{m}\in\mathcal{T}_{2},\\
\tilde{\mathbf{q}}\left(\mathbf{t}\right) & \geq & 0  ,
\end{eqnarray*}
for sufficiently small $\varepsilon_{1}\leq\varepsilon$ associated
with the kernel $\mathbf{g_{2}}$, and some constants $c_{1},c_{2}>0$. For
sufficiently large $\nu>0$ and constants $c_{a},c_{b}>0$, we also
have 
\begin{eqnarray*}
\left\Vert \mathbf{a}\right\Vert _{\infty}: & = & \underset{m}{\max}\left\vert a_{m}\right\vert \leq\frac{1}{\mathbf{g_2}(0,0)}+c_{a}\nu^{-3},\\
\left\Vert \mathbf{\tilde{b}}\right\Vert _{\infty}: & = & \underset{m}{\max}\left\vert b_{m}^{1}\right\vert ,\left\vert b_{m}^{2}\right\vert \leq c_{b}\nu^{-6}.
\end{eqnarray*}

\end{lem}
We present now the two-dimensional version of Lemma \ref{claim:qh}
without a proof. The proof relies on the same methodology as the one-dimensional
case. 
\begin{lem}
\label{lem:3.4_2}Let $\left\{ \mathcal{T}_{i,2}\right\} _{i=1}^{r}$  be a union of $r$ non-intersecting
sets  obeying $\mathcal{T}_{i,2}\in\mathcal{R}_{2}^{idx}\left(\nu\sigma,1\right)$
for all $i\in\left\{ 1,\dots,r\right\} $. For each set $\mathcal{T}_{i,2}$,
let $\tilde{\mathbf{q}}_\mathbf{i,2}[\mathbf{k}]:=\tilde{\mathbf{q}}_\mathbf{i,2}\left(\mathbf{k}/N\right)$,
$\mathbf{k}\in\mathbb{Z}^{2}$, be an associated function, where $\tilde{\mathbf{q}}_\mathbf{i,2}(\mathbf{t})$
is given in (\ref{eq:q2_form}). Then, for any sequence $\left\{ \alpha_{i}\right\} _{i=1}^{r}\in\left\{ 0,1\right\} $
we have for sufficiently large $\nu$, 
\begin{equation}
\sum_{{k}\in\mathbb{Z}^{2}}\prod_{i=1}^{r}\left(\tilde{\mathbf{q}}_\mathbf{i,2}\left[{k}\right]\right)^{\alpha_{i}}\mathbf{h}[{k}]\leq\tilde{C}_{2}(\mathbf{g},r)\delta+c^{*}\nu^{-6}\left\Vert \mathbf{h}\right\Vert _{1},\label{eq:lemma3.4_2}
\end{equation}
for some constants $c^{*}>0$ and $\tilde{C}_{2}(\mathbf{g_2},r)$ which depends
on the kernel $\mathbf{g_2}$ and the regularity parameter $r$. 
\end{lem}
Let $\mathbf{k}\in\mathbb{Z}^{2}$. Let us define the sets $\mathcal{N}_{2}:=\left\{ \mathbf{k}/N\thinspace:\thinspace \mathbf{h}[\mathbf{k}]<0\right\} $
and $\mathcal{N}_{2}^{C}:=\left\{ \mathbf{k}/N\thinspace:\thinspace \mathbf{h}[\mathbf{k}]\geq0\right\} $.
Throughout the proof, we use the notation of $\mathbf{k}\in\mathcal{N}_{2}$
and $\mathbf{k}\in\mathcal{N}_{2}^{C}$ to denote all $\mathbf{k}\in\mathbb{Z}^{2}$
so that $\mathbf{k}/N\in\mathcal{N}_{2}$ and $\mathbf{k}/N\in\mathcal{N}_{2}^{C}$,
respectively. By definition, $\mathcal{N}_{2}\in\mathcal{R}_{2}^{idx}\left(\nu\sigma,r\right)$ (see Definition \ref{def:2Dray})
and it can be presented as the union of non-intersecting subsets $\mathcal{N}_{2}=\cup_{i=1}^{r}\mathcal{N}_{i,2}$
where $\mathcal{N}_{i,2}\in\mathcal{R}_{2}^{idx}\left(\nu\sigma,1\right)$.
Therefore, for each subset $\mathcal{N}_{i,2}$ there exists an associated
function $\tilde{\mathbf{q}}_\mathbf{i,2}[\mathbf{k}]=\tilde{\mathbf{q}}_\mathbf{i,2}\left(\mathbf{k}/N\right)$
 given in Lemma \ref{lemma:q2}. As in the one-dimensional case, the proof relies on the following construction
\begin{equation}
\mathbf{q_{2}}[\mathbf{k}]:=\prod_{i=1}^{r}\left(1-\tilde{\mathbf{q}}_\mathbf{i,2}[\mathbf{k}]\right)-\rho,\label{eq:q2}
\end{equation}
for some constant $\rho>0$, to be defined later. This function satisfies
the following interpolation properties:
\begin{lem}
\label{lem:q_prod2} Suppose that 
\begin{equation*}
N\sigma>\max\left\{ \sqrt{\frac{c_{1}}{c_{2}}},\left(\varepsilon_{1}\right)^{-1},\left(\frac{1}{2}\right)^{\frac{1}{2r}}\sqrt{c_{1}}\right\},
\end{equation*}
where $\varepsilon_{1}$ is given in Lemma \ref{lemma:q2}. Let $\mathbf{q_{2}}$
be as in (\ref{eq:q2}) and let 
\begin{equation}
\rho\geq\frac{1}{2}\left(\frac{c_{1}}{\left(N\sigma\right)^{2}}\right)^{r}.\label{eq:rho2}
\end{equation}
Then, 
\begin{eqnarray*}
\mathbf{q_{2}}\left[\mathbf{k}_{m}\right] & = & -\rho,\quad\mathbf{k}_{m}\in\mathcal{N}_{2},\\
\mathbf{q_{2}}\left[\mathbf{k}\right] & \geq & \rho,\quad \mathbf{k}\in\mathcal{N}_{2}^{C},\\
\mathbf{q_{2}}\left[\mathbf{k}\right] & \leq & 1,\quad\mathbf{k}\in\mathbb{Z}^{2}.
\end{eqnarray*}
\end{lem}
\begin{proof}
Since $\mathcal{N}_{i,2}\in\mathcal{R}_{2}^{idx}\left(\nu\sigma,1\right),$
by Lemma \ref{lemma:q2} there exists for each subset $\mathcal{N}_{i,2}$
an associated function $\tilde{\mathbf{q}}_\mathbf{i,2}[{k}]=\tilde{\mathbf{q}}_\mathbf{i,2}(\mathbf{k}/N)$.
Consequently, for all $\mathbf{k}_{m}\in\mathcal{N}_{2}$ we obtain
\[
\mathbf{q_{2}}\left[\mathbf{k}_{m}\right]=\prod_{i=1}^{r}\left(1-\mathbf{q_{i,2}}\left[\mathbf{k}_{m}\right]\right)-\rho=-\rho.
\]
For $N\sigma\geq\max\left\{ \sqrt{\frac{c_{1}}{c_{2}}},\left(\varepsilon_{1}\right)^{-1}\right\} $
we get for all $\mathbf{k}\in\mathcal{N}_{2}^{C}$ 
\begin{eqnarray*}
\mathbf{q_{2}}[\mathbf{k}] & = & \prod_{i=1}^{r}\left(1-\mathbf{q_{i,2}}\left[\mathbf{k}\right]\right)-\rho\geq\left(\frac{c_{1}}{\left(N\sigma\right)^{2}}\right)^{r}-\rho.
\end{eqnarray*}
By setting 
\[
\rho:=\arg\min_{\mathbf{k}\in\mathcal{N}_{2}^{C}}\mathbf{q_{2}}[\mathbf{k}]\geq\frac{1}{2}\left(\frac{c_{1}}{\left(N\sigma\right)^{2}}\right)^{r},
\]
we conclude the proof. The condition $N\sigma>\left(\frac{1}{2}\right)^{\frac{1}{2r}}\sqrt{c_{1}}$
guarantees that $\rho<1$.
\end{proof}
Once we constructed the function $\mathbf{q_{2}}[\mathbf{k}]$, the proof follows
the one-dimensional case. By considering Lemmas \ref{eq:lemma3.4_2}
and \ref{lem:q_prod2} and using similar arguments to (\ref{eq:low-bound})
and (\ref{eq:upper-bound}), we conclude 
\begin{equation*}
\begin{split}
\rho\left\Vert \mathbf{h}\right\Vert _{1}&\leq\left\langle \mathbf{q},\mathbf{h}\right\rangle \leq\left(3^{r}-1\right)\tilde{C}_{2}(\mathbf{g}_2,r)\delta\\&+c^{*}\left(3^{r}-1\right)\nu^{-6}\left\Vert \mathbf{h}\right\Vert _{1}.
\end{split}
\end{equation*}
Using (\ref{eq:rho2}) we get for sufficiently large $\nu$ that 
\[
\left\Vert \mathbf{h}\right\Vert _{1}\leq C_{2}\left(\mathbf{g}_2,r\right)\left(N\sigma\right)^{2r}\delta,
\]
 for some constant ${C}_{2}(\mathbf{g}_2,r)$ which depends on the kernel
$\mathbf{g}_2$ and the Rayleigh regularity $r$.

\section{\label{sec:Numerical-Experiments}Numerical Experiments}

We conducted numerical experiments to validate the theoretical results of this paper.
The simulated signals were generated in two steps. First, random locations were
sequentially added to the signal\textquoteright s support in the interval
$[-1,1]$ with discretization step of 0.01, while keeping a fixed regularity
condition. Once the support was determined, the amplitudes were drawn
randomly from an i.i.d normal distribution with standard deviation
SD = 10. For positive signals, the amplitudes are taken to be the absolute values of the normal variables.

The experiments were conducted with the Cauchy kernel $\mathbf{g}(t)=\frac{1}{1+\left(\frac{t}{\sigma}\right)^{2}}$,
$\sigma=0.1$. We set the separation constant to be $\nu=0.5$, which
was evaluated in \cite{bendorySOP} to be  the minimal separation constant, guaranteeing the existence of
interpolating polynomial as in Lemma \ref{lemma:q}.
Figure \ref{fig:Example} presents an example for the estimation of
the signal (\ref{eq:x}) from (\ref{eq:model1}) with $r=2$. As can
be seen, the solution of the convex problem (\ref{eq:cvx}) detects
the support of the signal with high precision in a noisy environment of $27$ dB. Figure \ref{fig:2d}
presents an example for recovery  of a two-dimensional signal from a stream of  Cauchy kernels with $r=2$ and
$\nu=0.8$.

\begin{figure}
\begin{centering}

\includegraphics[scale=0.4]{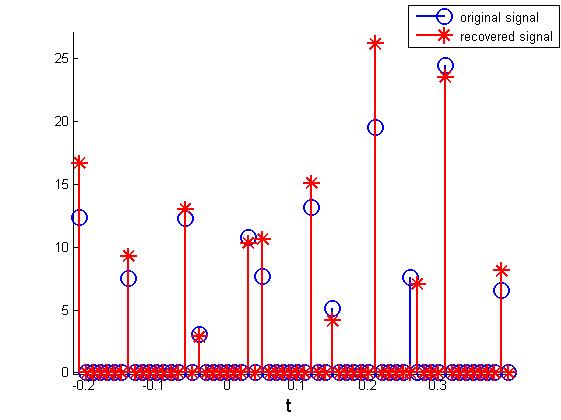}
\end{centering}

\protect\caption{\label{fig:Example}Example for the recovery of a signal of the form of (\ref{eq:x}) from stream of  Cauchy kernels
with $\sigma=0.1$, Rayleigh regularity of $r=2$, separation constant
of $\nu=0.5$ and noise level of $\delta=75$ (SNR=27dB).}
\end{figure}

\begin{figure}
\begin{centering}
\includegraphics[scale=0.4]{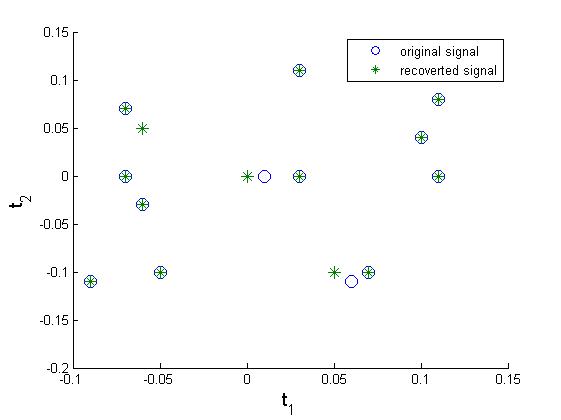}
\par\end{centering}

\protect\caption{\label{fig:2d}An example for the recovery of a two-dimensional signal of
the form (\ref{eq:x2}) from the measurements (\ref{eq:model2}),
with $r=2$, $\delta=400$ and $\nu=0.8$. The figure presents merely the
locations (support) of the original and the recovered signals.}
\end{figure}

Figure \ref{fig:mean_err} shows the localization error as a function
of the noise level $\delta$. To clarify, by localization error we
 mean the distance between the support of the original
signal and the support of the recovered signal. Figure \ref{fig:positive-negative}
compares the localization error for positive signals 
and general real signals (i.e. not necessarily positive) from stream of Cauchy pulses. For general signals,
we solved a standard $\ell_{1}$ minimization problem as in \cite{bendorySOP},
which is the same problem as (\ref{eq:cvx}) without the positivity
constraint $\mathbf{x}\geq0$. Plainly, the  localization error of positive signals is significantly smaller
than the error of  general signals. Figure \ref{fig:err_r}
shows that the error grows approximately linearly with the noise level
$\delta$ and increases with $r$.

\begin{figure*}
\begin{centering}

\subfloat[\label{fig:positive-negative}Mean localization error of positive signal
 and general signal (not necessarily positive coefficients) with
 $r=2.$]{\includegraphics[scale=0.4]{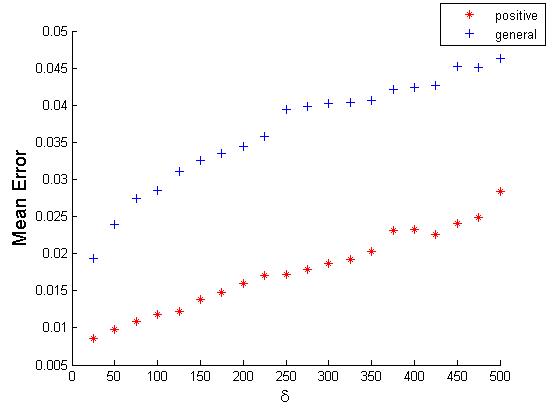}} \hspace{10pt} 
\subfloat[\label{fig:err_r}Mean localization error of positive signals for $r=2,3,4$.]{\includegraphics[scale=0.4]{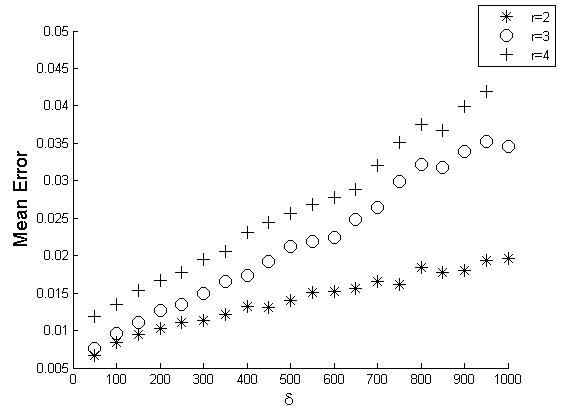}}

\par\end{centering}
\protect\caption{\label{fig:mean_err}Mean localization error from a stream of Cauchy pulses  as a function of the
noise level $\delta$. For each value of $\delta,$ 50 experiments
were conducted. }
\end{figure*}

\section{\label{sec:Conclusions}Conclusions}

In this paper, we have shown that a standard convex optimization program
can robustly recover the sets of delays and positive amplitudes from a stream of pulses. The recovery error is proportional to the noise level and grows exponentially with the density of signal's support, which is defined by the notion of Rayleigh regularity. The error also depends on the localization properties of the kernel. 
 In contrast to general stream of pulses
model as discussed in \cite{bendorySOP}, no separation is needed and the signal's support may be clustered.
It is of great interest to examine the theoretical results we have
derived on real applications, such as detection and tracking tasks in single-molecule
microscopy.

We have shown explicitly that our technique holds true for one and two
 dimensional signals. We strongly believe that similar results hold
for higher-dimension problems. Our results rely on the existence of
interpolating functions which were constructed in a previous work
\cite{bendorySOP}. Extension of the results of \cite{bendorySOP} to
higher dimensions will imply immediately the extension of our results
to higher dimensions as well.

In \cite{SOP_US}, it was shown that for general signals that satisfy
the separation condition (\ref{eq:1dseparation}), the solution of
a convex program results in a localization error of order 
$\sqrt{\delta}$. Namely, the support of the estimated signal is clustered around the support of the sought signal. It would be interesting to examine whether such a phenomenon exists in the positive case as well.

\section*{Acknowledgement}

The author is grateful to Prof Arie Feuer and Prof Shai Dekel for
their comments and support and to Veniamin Morgenshtern for helpful
discussions about \cite{morgenshtern2014stable}.


\bibliographystyle{ieeetr}


\end{document}